\documentclass{JCIN18}
\usepackage{bm}
\usepackage{graphicx}
\usepackage{subfig}

\newcommand{\spaceblank}{\vskip 4mm}

\newtheorem{Definition}{Definition}
\newtheorem{Problem}{Problem}
\newtheorem{Lemma}{Lemma}

\newtheorem{Algorithm}{Algorithm}
\newtheorem{Policy}{Policy}

\newtheorem{Remark}{Remark}

\begin{document}

\ArticleType{Research paper}
\Year{2019}
%\Acknowl{This work was supported by the National Natural Science Foundation of China (No. 61771232).}

\title{Cooperative Job Dispatching in Edge Computing Network with Unpredictable Uploading Delay}

\author[1]{Bojie Lv}
\author[1,2]{Yuncong Hong}
\author[3]{Haisheng Tan}
\author[2]{Zhenhua Han}
\author[1]{Rui Wang}

\address[1]{Department of Electrical and Electronic Engineering, The Southern University of Science and Technology (SUSTech)}
\address[2]{Department of Computer Science, The University of Hong Kong (HKU)}
\address[3]{School of Computer Science and Technology, The University of Science and Technology of China (USTC)}

\abstract{In this paper, the cooperative jobs dispatching problem in an edge computing network with multiple access points (APs) and edge servers is considered. Due to the uncertain traffic in the network between APs and edge servers, the job uploading delay can not be predicted accurately. Specifically, the job arrivals at the APs, the job uploading delay from APs to edge servers and the job computation time at the edge servers are all modeled as random variables. Since each job dispatching decision will affect the system state in the future, we formulate the joint optimization of jobs dispatching at all the APs and all the scheduling time slots as an infinite-horizon Markov decision process (MDP). The minimization objective is a discounted measurement of the average processing time per job, including the uploading delay, the waiting time and the computation time at the edge servers. In this problem,
	 the approximate MDP should be adopted to address the curse of dimensionality. Conventional low-complexity approximate solution of MDP is usually hard to predict the performance analytically. In this paper, a novel approximate MDP solution framework is proposed via one-step policy iteration over a baseline policy, where the analytical performance bound can be obtained. Moreover, since the expression of the approximate value function is derived, the value iteration in conventional methods can be eliminated, which can essentially reduce the computation complexity. It is shown by simulations that proposed low-complexity algorithm has significantly better performance than various benchmark schemes.}


\maketitle

\section{Introduction}
A number of emerging mobile applications, such as face recognition, speech recognition and high-definition video rendering, are computation-intensive and delay-sensitive. Because of the limited computation resource and battery capacity, it is promising for the mobile devices to upload their computation-intensive jobs to edge servers with much more powerful computation capability. In this paper, we focus on the jobs dispatching optimization between multiple access points (APs) and edge servers via a network with random job uploading delay. Unlike cellular communications, the uploading delay between APs and edge servers is hard to control due to unpredictable traffics in the network. We shall address the dispatching optimization in this scenario via a novel approximation Markov decision process (MDP) method, whose performance can be analytically bounded.

\subsection{Related Works}
The scheduling algorithm design for edge computing systems has attracted tremendous research attentions.
There have been a number of works considering the radio resource management for mobile edge computing systems. For example, the authors in \cite{Junzhang2016} minimized the average energy consumption in a single-user system via Lyapunov optimization approach. In \cite{KBHuang2015}, the authors derived the closed-form expressions of job uploading decisions and the allocation of computation and radio resources in a single-user system powered by wireless energy transfer. 
Considering the dynamic of CPU state (busy or idle) at the edge server, the authors in \cite{KBHuang2019} proposed a dynamic job offloading
	algorithm to minimize the average energy consumption in single-user system via finite-horizon MDP.
There are also a significant number of works considering the edge computing scenarios with multiple mobile users and single edge server. For example, the radio and computation resources allocation to guarantee user fairness and delay constraint in a multi-user system was considered in \cite{Du2018}. The authors in \cite{KBHuang2016} proposed an optimal threshold-based uploading algorithm for mobile users. In order to minimize a weighted summation of total energy consumption and uploading delay, the authors in \cite{XuChen2016} proposed a distributed job uploading algorithm based on game theory. All these works considered the scheduling algorithm design in a wireless cell with single edge server. Moreover, the authors considered joint optimization of service caching and job uploading with multiple edge servers in \cite{jieXu2018}. 
	Due to the limited storage space, edge servers can not process all job types.  An online and decentralized scheduling algorithm is proposed to minimize the computation delay under total computation energy constraint.
All the above works consider the scheduling of wireless transmission from mobile users to APs. In fact, there is also dispatching issue between APs and edge servers. For example, in a computer network, the job uploading delay from users or APs to edge servers is not negligible \cite{Tan2017,liang2017}. Moreover, the delay may be unpredictable, as it may be jammed by other traffics.

There are also some works on the jobs dispatching design 
 in computer networks. For example, without any job arrival information, the authors in \cite{Tan2017} designed an online algorithm for jobs dispatching in edge computing systems to achieve a good competitive ratio. Given a consistent network transmission delay, the edge server placement and static jobs dispatching are jointly optimized in \cite{liang2017} to minimize overall job uploading delay. The uploading delay and job computation time are assumed to be constant in \cite{Tan2017,liang2017}. 
In practice, however, the network traffic between APs and edge servers is usually complicated.
	The job uploading path is usually established with dynamic routing algorithm, unpredictable backlogs on the routers, and burst network flows from other services \cite{liang2015}.
As a result, it may be impractical to assume that the uploading delay is deterministic in edge computing systems, and new algorithm design framework addressing random uploading delay becomes necessary.

\begin{figure*}
	\centering
	\includegraphics[width=0.60\textwidth]{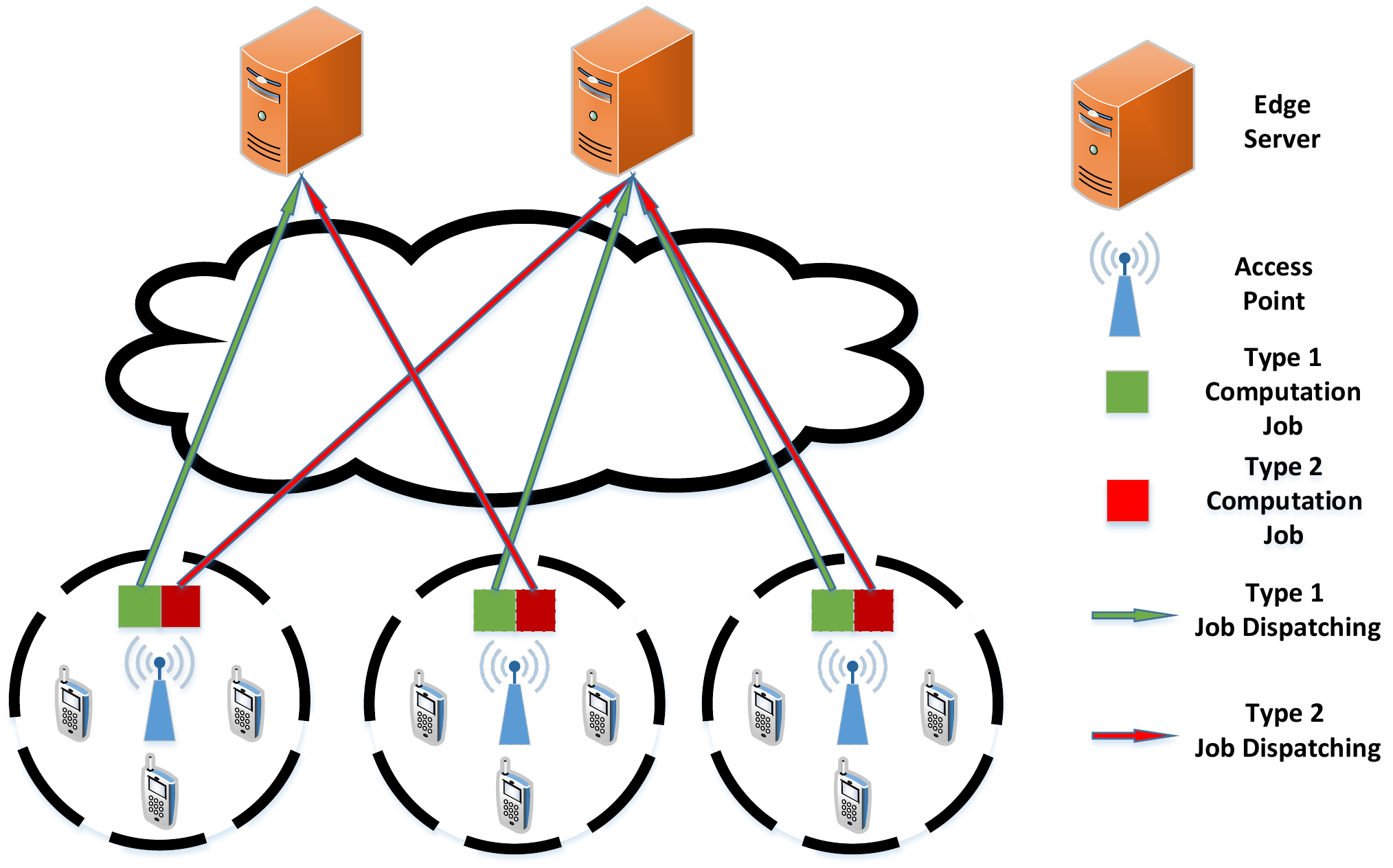}
	\caption{Illustration of the edge computing network.}
	\label{fig:system}
\end{figure*}

Finally, MDP is a powerful tool for resource allocation of communication networks with random transition of system state. For example,  
infinite-horizon average cost MDP have been used in delay-aware radio resource management \cite{Ruiwang2011,Cui2012,Ruiwang2013}. Joint optimization of file placement and delivery in cache-assisted wireless networks can be solved via finite-horizon MDP \cite{Lv2018-icc,Lv2018-gc,Lv2019}. Various value function approximation methods have been used in \cite{Ruiwang2011,Cui2012,Ruiwang2013,Lv2018-icc,Lv2018-gc,Lv2019} 
to address the {\em curse of dimensionality}. However, there is no analytical performance bound on the proposed approximation algorithms.

\subsection{Our Contributions}
In this paper, we would like to shed some lights on the above issue by optimizing the jobs dispatching from multiple APs to multiple edge servers via a network with unpredictable uploading delay. Specifically, we consider the {\em unrelated machines model} so that different edge servers may have different processing capability on each job type. 
The job arrivals, 
 job uploading delay from APs to edge servers and the job computation time at the edge servers are all modeled via random variables.
 Our contributions in this jobs dispatching scenario are summarized below.
\begin{itemize}
	\item We formulate the joint optimization of jobs dispatching in all the APs and time slots as an infinite-horizon MDP, where the minimization objective is a discounted measurement of job processing time, including the uploading delay, the waiting time and the computation time  at
		edge servers. The issue of random uploading delay and computation time is addressed via the state transition distribution of MDP formulation.
	\item Conventional MDP problems suffer from {\em curse of dimensionality}.
	In order to address this issue, a novel approach of value function approximation is proposed for the above infinite-horizon MDP with discounted cost, where the expressions of approximated value function is derived. Hence, the complicated value iteration is avoided. Moreover, with this new approach, the performance of the proposed dispatching algorithm can be analytically bounded. 
\end{itemize}

The remainder of this paper is organized as follows. The system model is
presented in Section 2. Problem formulation and low-complexity scheduling are illustrated in Section 3 and Section 4, respectively.
In Section 5, numerical simulations are conducted. Finally, the work is
concluded in Section 6.

	We use the following notation throughout this paper: bold lowercase $\textbf{a}$ is used to denote column vectors, bold uppercase $\textbf{A}$ is used to denote matrices, non-bold letters a, $A$ are used to denote scalar values, and caligraphic letters $\mathcal{A}$ to denote sets. Using these notations,  $\lceil a \rceil$ is the smallest integer not
	smaller than $a$; $\left[\textbf{A}\right]_{i,j}$ and $\textbf{A}^\mathsf{T}$ denote the $(i,j)$-th element of $\textbf{A}$ and transpose respectively. $\mathbf{I}$ denotes identity matrix. $\mathbb{G}(p)$ denotes the geometric distribution with parameter $p$; $\mathbb{B}(n,p)$ denotes binomial distribution with parameters $n$ and $p$; $\mathbb{E}\left[.\right]$ denotes an expectation operator; $\mathrm{I}(.)$ denotes an indicator function; $\mathbb{R}^{M\times N}$ denotes spaces of $M\times N$ matrices with real entries.\\

\section{System Model}

In this section, we introduce the model of the edge computing system considered in this paper, including the statistical models of job arrival, uploading and computation.

\subsection{Network Model}
We consider an edge computing system with $K$ access points (APs) and $M$ edge servers, which are connected in a network as illustrated in Fig.\ref{fig:system}. The sets of APs and edge servers are denoted as $\mathcal{K}\triangleq\{1,...,K\}$ and $ \mathcal{M}\triangleq\{1,...,M\}$, respectively.
Each AP collects the computation jobs from the mobile users within its service area, and uploads each job to one of the edge servers. Without loss of generality, it is assumed that there are $J$ types of computation jobs supported in this system, which are denoted via the set $\mathcal{J} \triangleq \{1,\dots,J\}$. The edge severs may have different processing capability on different job types. The APs and edge servers may be deployed in an open network (e.g., metropolitan area network) with other traffics (e.g., video streaming and file delivery).
	It is shown in a number of existing literature \cite{Tan2017,liang2017} that the job uploading delay is not negligible compared with the computation time. Moreover, due to the randomness of network traffics,	
	the job uploading delay is assumed to be random. In this paper, we shall optimize the computation edge server for each job type at the APs, according to the distribution of job uploading delay, the queuing status and the job processing capability of edge servers.

The time axis is organized by time slots in order to facilitate the dispatcher design. The job arrival in each time slot is modelled via Bernoulli distribution. Specifically, the arrivals of the $j$-th job type at the $k$-th AP in different time slots are independent and identically distributed (i.i.d.) Bernoulli random variables, and the arrival probability is denoted as $\lambda_{k,j}$ ($\forall k\in\mathcal{K}, j\in\mathcal{K}$). Let $A_{k,j}(t)\in\{0,1\}$ be the indicator of job arrival, where $A_{k,j}(t)=1$ means one job of the  $j$-th type arrives at the $k$-th AP in the $t$-th time slot, and $A_{k,j}(t)=0$ means otherwise. Hence,
\begin{equation}
	\Pr \bigg( A_{k,j}(t)=1 \bigg) = \lambda_{k,j}, \ \ \forall k,j,t.
\end{equation}

At the beginning of each time slot, APs dispatch each type of jobs arrived in the previous time slot to one edge server. Thus, the APs make decisions on the mapping from job types to edge severs in each time slot. We shall refer to these decisions in each time slot as dispatching actions. Let $\omega_{k,j}(t)\in\mathcal{M}$ denotes the index of edge server, to which the $k$-th AP dispatches the job of the $j$-th type in the $t$-th time slot. The dispatching action of the system in the t-th time slot can be represented as
\begin{align*}
\{\omega_{k,j}(t)|\forall k\in\mathcal{K},\forall j\in\mathcal{J}\}.
\end{align*}

 Different types of jobs may have different distributions on the input data size. Moreover, the network between APs and edge servers may be jammed by other traffics. The job uploading delay from one AP to one edge server cannot be predicted accurately by APs. Instead, it is assumed in this paper that the uploading delay follows independent geometric distribution. Denote the geometric delay distribution of the $j$-th job type from the $k$-th AP to the $m$-th edge server as $\mathbb G\left(1/\bar{U}_{k,j}^{m}\right)$, where $\bar{U}_{k,j}^{m}$ is the expectation of the distribution. 

\begin{Remark}[Memoryless Uploading Delay Distribution] \emph{
 The  geometric distribution has the memoryless property. For example, let $U^m_{k,j}(t)$ be the uploading delay of the job of the $j$-th type which is dispatched from the $k$-th AP to the $m$-th edge sever in the $t$-th time slot. Then,
 $\forall n>0, s>0, t, k\in\mathcal{K},j\in\mathcal{J},m\in\mathcal{M}$, 
 \begin{align*}
 \Pr\bigg({U}_{k,j}^{m}(t)> n+s\bigg|{U}_{k,j}^{m}(t)>n\bigg)=\Pr\bigg({U}_{k,j}^{m}(t)>s\bigg).
 \end{align*}
 As a result, the statistics of job arrivals at the edge servers depend only on the number of jobs which are being delivered from APs to edge servers. It is not necessary for the APs to record the number of time slots for which these jobs has been delivered from the AP.
  However, our proposed algorithm is not limited to the geometric delay distribution. It
   can be easily extended to the scenarios that the job uploading delay follows other distributions. We use the geometric distribution as it can simplify the notation system.}
\end{Remark}

Let $N_{k,j}^m(t)$ be the number of the jobs of the $j$-th type, which is being uploaded from the $k$-th AP to the $m$-th edge server at the beginning of the $t$-th time slot, $D_{k,j}^m(t)\in\{0,1,\dots,N_{k,j}^m(t)\}$ be the number of the jobs of the $j$-th type which arrive at the $m$-th edge server from the $k$-th AP in the $t$-th time slot, respectively.
As a remark notice that the data of the jobs in $N_{k,j}^m(t)$ have not arrived at the $m$-th edge server by the beginning of the $t$-th time slot. Due to the random uploading delay, some of these jobs may arrive during the $t$-th time slot, which are measured by $D_{k,j}^m(t)$. Hence, $D_{k,j}^m(t)$ follows binomial distribution with expectation $N_{k,j}^m(t)/\bar{U}_{k,j}^{m}$, i.e., $D_{k,j}^m(t)\sim \mathbb{B}(N_{k,j}^m(t),1/\bar{U}_{k,j}^{m})$,
and the probability mass function (PMF) of $D_{k,j}^m(t)$ is given by
\begin{align}
\Pr\left(D_{k,j}^m(t)=n\right)&=\binom{N_{k,j}^m(t)}{n}\left(\frac{1}{\bar{U}_{k,j}^{m}}\right)^n\left(1-\frac{1}{\bar{U}_{k,j}^{m}}\right)^{N_{k,j}^m(t)-n}, \nonumber\\ 
\forall n&=0,1,\dots,
N_{k,j}^m(t).
\end{align}
Hence, given job arrival process $A_{k,j}(t)$ and jobs dispatching decision $\omega_{k,j}(t)$,  the dynamics of $N_{k,j}^m(t+1)$ ($\forall t, k\in \mathcal{K}, m\in \mathcal{M}, j \in \mathcal{J}$)  can be expressed as
\begin{align}
N_{k,j}^m(t+1)=&N_{k,j}^m(t)+A_{k,j}(t)\mathrm{I}\Big(\omega_{k,j}(t)=m\Big)-D_{k,j}^m(t).
\end{align}
In the above equation, $\mathrm{I}(\mathcal{E})$ is the indicator function, whose value is $1$ when the event $\mathcal{E}$ is true and $0$ otherwise.

\subsection{Computation Model}
There are $J$ virtual machines (VMs) on each edge server for the computation of $J$ job types, respectively. For each type, the uploaded jobs are computed in a first-come-first-serve (FCFS) manner.
Hence, a processing queue with maximum $L_\text{max}$ jobs is established for each VM, and the first job is computed. The arrival jobs will be discarded when the processing queue is full.
Denote $L_{m,j}(t)\in\{0,1,\dots,L_{\text{max}}\}$ as the number of the jobs of the $j$-th type at the $m$-th edge server at the beginning of the $t$-th time slot.

We adopt the \emph{unrelated machines} assumption as in \cite{Tan2017} for job computation on edge servers. Specifically, it is assumed that different types of jobs have different distributions of computation time at each edge server.
We denote $f_{m,j}(x)$ as the PMF of computation time distribution of the j-th job type at the m-th edge server ($\forall m\in\mathcal{M}, j \in \mathcal{J}$).
 Let $\eta_{m,j}(t)\in\{0,1,\dots,\eta_{\text{max}}\}$ be the remaining computation time slots of the first job at the $j$-th VM at the beginning of the $t$-th time slot, where $\eta_{\text{max}}$ denotes the maximum number of computation time slots for each job. Then dynamics of  $\{\eta_{m,j}(t)|\forall t\}$ are summarized below: 
 \begin{itemize}
	\item When the $\eta_{m,j}(t)>1$,  $\eta_{m,j}(t+1)=\eta_{m,j}(t)-1$;
	\item When the the computation of the first job is finished in the $t$-th time slot ($ \eta_{m,j}(t)= \{0,1\} $) and there are no job in the processing queue ($ L_{m,j}(t)=0 $), ${\eta}_{m,j}(t+1)=0$;
	\item When the the computation of the current job is finished in or before the $t$-th time slot ($ \eta_{m,j}(t)= 1 $) and $ L_{m,j}(t)>0 $, the distribution of $\eta_{m,j}(t+1)$ is given by 
	\begin{align}
	\Pr\bigg(\eta_{m,j}(t+1)=x\bigg)=f_{m,j}(x), \forall t, x\in \{0,1,\dots,L_{\text{max}}\}.
	\end{align}
\end{itemize}
Moreover, the dynamics of $L_{m,j}(t)$ can be expressed as
\begin{align}
&L_{m,j}(t+1) \nonumber\\ 
&\quad=\min\bigg\{L_{m,j}(t)-\mathrm{I}\Big(\eta_{m,j}(t)=1\Big)+\sum_{\forall k \in \mathcal{K}}D_{k,j}^m(t),\quad L_{\text{max}}\bigg\},\nonumber\\ 
&\quad\quad\forall t, m\in\mathcal{M},j\in\mathcal{J}.
\end{align}
 
In the remaining of this work, we shall refer to $$Q_{m,j}(t)\triangleq\bigg(L_{m,j}(t),\eta_{m,j}(t)\bigg)$$ as the queuing state information (QSI) of the $j$-th type job at the $m$-th edge server at the beginning of the $t$-th time slot.

\section{Infinite-Horizon MDP Formulation}

Since the jobs dispatching in one time slot will affect the system status (e.g., QSI of the edge servers) of the following time slots.
 The joint optimization of jobs dispatching in all the time slots is necessary. In this section, we shall formulate such joint optimization as a MDP. 

\subsection{System State and Scheduling Policy}

We first define the system state $\mathbf{S}$ and scheduling policy $\Omega$ as follows. 

\begin{Definition}[System State] \em
	At the beginning of the $t$-th time slot, the system state of the $j$-th job type is represented as  $\mathbf{S}_j(t)\triangleq(\mathbf{N}_j(t),\mathbf{Q}_j(t))$, which consists of
\begin{itemize}

		\item The number of jobs being uploaded:
	\begin{align}
	\mathbf{N}_j(t)\triangleq\{N_{k,j}^m(t)|\forall k \in \mathcal{K}, \forall m \in \mathcal{M} \};
	\end{align}

	\item Queuing state information (QSI) of the edge servers:
	\begin{align}
	\mathbf{Q}_j(t)\triangleq \{Q_{m,j}(t)|\forall m \in \mathcal{M}\}.
	\end{align}
	
\end{itemize}
		
	Moreover, the aggregation of system state of all the type of jobs is referred to as the system state $\mathbf{S}(t)$, i.e., $\mathbf{S}(t) \triangleq \{\mathbf{S}_j(t) | \forall j\in\mathcal{J} \}$.

\end{Definition}

It is assumed that all the APs and edge servers will broadcast their latest status at the end of every time slot, and the APs are able to collect the complete system state at the beginning of each time slot (i.e., $\mathbf S (t)$ at the beginning of the $t$-th time slot), so that the decision on jobs dispatching can be made accordingly. In this work, we ignore the delay of system state broadcasting, as the message size is small. Hence, the jobs dispatching policy is defined below.

\begin{Definition}[Jobs Dispatching Policy] \em
	 In the $t$-th time slot, the dispatching policy of the $j$-th job type, denoted as $\Omega_j$, is a mapping from system state $\mathbf{S}(t)$ to the jobs dispatching action  $\{ \omega_{k,j}(t)|\forall k\in\mathcal{K} \}$, i.e.,
		\begin{align}
		\Omega_{j}(\mathbf{S}(t))\triangleq\{\omega_{k,j}(t)|\forall k\in\mathcal{K}\}, \forall t.
		\end{align}
	Moreover, the aggregation of dispatching polices of all the job types is referred to as the system dispatching policy $\Omega$, i.e., $\Omega \triangleq \{\Omega_j | \forall j\in\mathcal{J} \}$.
\end{Definition}

\subsection{Problem Formulation}
According to the Little's law, the average processing time per job of the edge computing system, measuring the number of time slots from job arrival to the completeness of computation, is proportional to the average number of jobs in the system. In this paper, we use the discounted summation of job numbers in all the time slots as the approximation of average processing time. Specifically, we first define the following weighted sum of the job number and job overflow penalty as the system cost at the $t$-th time slot. 
\begin{align}
&g\bigg(\mathbf{S}(t),\Omega(\mathbf{S}(t))\bigg)  
\triangleq\sum_{j\in\mathcal{J}}\bigg\{\nonumber\\
&\underbrace{\sum_{k\in\mathcal{K}}\sum_{m\in\mathcal{M}}N_{k,j}^m(t) +\sum_{m\in\mathcal{M}}\bigg(L_{m,j}(t) 
+\beta\mathrm{I}\Big(L_{m,j}(t)=L_{\text{max}}\Big)\bigg)}_{g_j(\mathbf{S}_j(t),\Omega_j(\mathbf{S}_{j}(t)))}\bigg\},
\end{align}
where $\beta$ is a weight, and $g_j(\mathbf{S}_j(t),\Omega_j(\mathbf{S}_{j}(t)))$ denotes the system cost of $j$-th type in the $t$-th time slot. The overall system cost of all the time slots with the initial system state $\mathbf{S}$ is then given by
\begin{align}
\bar{G}(\Omega,\mathbf{S})\triangleq \lim\limits_{T\to \infty}\mathbb{E}_{\{\mathbf{S}(t)|\forall t\}}^{\Omega}\bigg[\sum_{t=1}^{T} \gamma^{t-1}g\bigg(\mathbf{S}(t),\Omega(\mathbf{S}(t))\bigg)\bigg|\mathbf{S}(1)=\mathbf{S} \bigg],
\end{align}
where $\mathbb{E}_{\{\mathbf{S}(t)|\forall t\}}^{\Omega}[.]$ denotes the expectation with respect to all possible system states in the future given dispatching policy $\Omega$, and $\gamma$ is the discount factor. As a result, the cooperative jobs dispatching design can be formulated as the following infinite-horizon MDP.
\begin{Problem}[Cooperative Jobs Dispatching Problem]\label{Pro:main} \em
	\begin{align}
			\Omega^*&=\mathop{\arg\min}_{\Omega} \bar{G}(\Omega,\mathbf{S}) .
	\end{align}

	\end{Problem}

The optimal policy of Problem \ref{Pro:main} can be obtained by solving the following Bellman's equations \cite{DP}, .
\begin{align}\label{eqn:bellman}
V(\mathbf{S}(t))=&\min_{\Omega(\mathbf{S}(t))} g\bigg(\mathbf{S}(t),\Omega(\mathbf{S}(t))\bigg)\nonumber\\
&+\gamma\sum_{\mathbf{S}(t+1)}\Pr\bigg(\mathbf{S}(t+1)\bigg|\mathbf{S}(t),\Omega(\mathbf{S}(t))\bigg)\nonumber\\
&\quad\quad\quad\quad\quad\quad\times V(\mathbf{S}(t+1)),
\forall \mathbf{S}(t),
\end{align}
where $V(\cdot)$ denotes the value function of the optimal policy $\Omega^*$. It is proven in \cite{DP} that, $V(\mathbf S)$ represents the average system cost with initial system state $\mathbf S$ and optimal scheduling policy $\Omega^{*}$, i.e.,
\begin{eqnarray}
	V(\mathbf S) = \lim\limits_{T\to \infty}\mathbb{E}_{\{\mathbf{S}(t)|\forall t\}}^{\Omega^{*}}\bigg[\sum_{t=1}^{T} \gamma^{t-1}g\bigg(\mathbf{S}(t),\Omega^{*}(\mathbf{S}(t))\bigg)\bigg|\mathbf{S}(1)=\mathbf{S} \bigg]. \nonumber
\end{eqnarray}
	The system state transition probability can be written as
	\begin{align}\label{eqn:decouple}
	&\Pr\bigg(\mathbf{S}(t+1)\bigg|\mathbf{S}(t),\Omega(\mathbf{S}(t))\bigg) \nonumber\\
	=&\prod_{j\mathcal{J}}\Pr\bigg(\mathbf{S}_j(t+1)\bigg|\mathbf{S}_j(t),\Omega_j(\mathbf{S}_j(t))\bigg) \nonumber\\
		=&\prod_{j\in \mathcal{J},\forall k \in \mathcal{K}, \forall m \in \mathcal{M}}\Pr\bigg(N_{k,j}^m(t+1)\bigg|\mathbf{S}_j(t),\Omega_j(\mathbf{S}_j(t))\bigg)\nonumber\\
	&\times\prod_{j\in \mathcal{J},\forall m \in \mathcal{M}}\Pr\bigg(Q_{m,j}(t+1)\bigg|\mathbf{S}_j(t),\Omega_j(\mathbf{S}_j(t))\bigg).
	\end{align}

Generally speaking, the standard value iteration can be used to solve the value function $V(.)$ for all possible system states, and the optimal policy denoted as $\Omega^{*}$, can be derived by solving the minimization problem of the right-hand-side of the  Bellman's equations in \eqref{eqn:bellman}. In our problem, however, the conventional value iteration is intractable due to the tremendous state space. For example, the number of system states grows exponentially with respect to the number of APs and edge servers. Hence, a novel low-complexity sub-optimal solution is proposed in the following section, whose performance can be bounded analytically. As a remark notice that it is difficult to obtain an analytical bound for the existing approximate MDP solution methods as in \cite{Ruiwang2011,Cui2012,Ruiwang2013}.

\section{Low-Complexity Scheduling Policy}
In this section, we first introduce a heuristic scheduling
policy as the baseline policy, whose value functions are derived
analytically. Then, the proposed low-complexity sub-optimal
policy can be obtained via the above value function and one-step
policy iteration. The derived value function of the baseline
policy becomes the cost upper bound of the proposed sub-optimal
policy.

\subsection{Baseline Scheduling Policy}
The baseline scheduling policy with fixed dispatching action is elaborated below.
\begin{Policy}[Baseline Scheduling Policy $\Pi$]\label{Pol:baseline} \em The following dispatching policy $\Pi$ is adopted as the baseline policy.
\begin{align}
\Pi\triangleq\{\omega_{k,j}(t)=\omega_{k,j}^{\Pi}|\forall t,k,j\},
\end{align}
where $\omega_{k,j}^{\Pi}\in\mathcal{M}$ denotes the index of the fixed edge server for the processing of the $j$-th job type from the $k$-th AP.
\end{Policy}

Given the system state $\mathbf{S}$ in the first time slot, the value function of policy $\Pi$ is defined as
\begin{align}
V_{\Pi}(\mathbf{S})\triangleq \lim\limits_{T\to \infty}\mathbb{E}_{\{\mathbf{S}(t)|\forall t\}}^{\Pi}\bigg[\sum_{t=1}^{T} \gamma^{t-1}g(\mathbf{S}(t),\Pi)\bigg|\mathbf{S}(1)=\mathbf{S} \bigg].
\end{align}
In order to derive its analytical expression, we let 
\begin{align}
{d}_{k,j,m}^{\text{AP}}(N_{k,j}^m) \triangleq \lim\limits_{T\to \infty} \mathbb{E}^{\Pi}_{\{N_{k,j}^{m}(t)|\forall t\}}\bigg[\sum_{t=1}^{T}\gamma^{t-1}N_{k,j}^{m}(t)\bigg|N_{k,j}^{m}(1)=N_{k,j}^{m}\bigg]
\end{align}
be the average cost raised by the jobs of the $j$-th type which is being uploaded from $k$-th AP to $m$-th server, and 
\begin{align}
 {d}_{m,j}^{\text{ES}}(\mathbf{S}_j)\triangleq &\lim\limits_{T\to \infty} \mathbb{E}^{\Pi}_{\{N_{k,j}^{m}(t)|\forall t\}}\bigg[\sum_{t=1}^{T}\gamma^{t-1}L_{m,j}(t)\nonumber\\
 &+\beta\mathrm{I}(L_{m,j}(t)=L_{\text{max}})\bigg|\mathbf{S}_j(1)=\mathbf{S}_j\bigg]
\end{align}
 be the average cost cost raised by jobs of the $j$-th type at the $m$-th edge server. $V_{\Pi}(\mathbf{S})$ can be
written as 
\begin{align}\label{eqn:V_baseline}
V_{\Pi}(\mathbf{S})=\sum_{j\in\mathcal{J}}\bigg(\underbrace{\sum_{k\in\mathcal{K}}\sum_{m\in\mathcal{M}}{d}_{k,j,m}^{\text{AP}}(N_{k,j}^{m})+\sum_{m\in\mathcal{M}}{d}_{m,j}^{\text{ES}}(\mathbf{S}_j)}_{W_{j}(\mathbf{S}_j)}\bigg),
\end{align}
where  the expressions of  ${d}_{k,j,m}^{\text{AP}}(.)$ and ${d}_{m,j}^{\text{ES}}(.)$ are given by following two lemmas respectively.

\begin{table*}
	\centering  % 显示位置为中间
	\caption{ENTRIES OF MATRIX $\mathbf{M}_{k,j,m}$}  % 表格标题 
	\label{table:M_{k,j,m}}  % 用于索引表格的标签
	%字母的个数对应列数，|代表分割线
	% l代表左对齐，c代表居中，r代表右对齐
	\begin{tabular}{|c|c|c|}  
		\hline  % 表格的横线
		& & \\[-6pt]  %可以避免文字偏上来调整文字与上边界的距离
		$q$&$p$&$[\mathbf{M}_{k,j,m}]_{q,p}$ \\  % 表格中的内容，用&分开，\\表示下一行
		\hline
		& &  \\[-6pt]  %可以避免文字偏上 
		$0$&$0$&$1-\lambda_{k,j}$ \\
		\hline
		& &  \\[-6pt]  %可以避免文字偏上 
		$0$&$1$&$\lambda_{k,j}$ \\
		\hline
		& &  \\[-6pt] 
		$0$&$2,\dots,N_{\max}$&$0$ \\
		\hline
		& &  \\[-6pt] 
		$a\in\{1,\dots,N_{\max}-1\}$&$b\in\{0,\dots,a\}$&$(1-\lambda_{k,j})\binom{a}{a-b}(\frac{1}{\bar{U}_{k,j}^{m}})^{a-b}(1-\frac{1}{\bar{U}_{k,j}^{m}})^{b}+\lambda_{k,j}\binom{a}{a-b+1}(\frac{1}{\bar{U}_{k,j}^{m}})^{a-b+1}(1-\frac{1}{\bar{U}_{k,j}^{m}})^{b-1}$ \\
		\hline
		& &  \\[-6pt] 
		$a\in\{1,\dots,N_{\max}-1\}$&$a+1$&$\lambda_{k,j}(1-\frac{1}{\bar{U}_{k,j}^{m}})^{a}$ \\
		\hline
		& &  \\[-6pt] 
		$a\in\{1,\dots,N_{\max}-1\}$&$b\in\{a+2,\dots,N_{\max}\}$&$0$ \\
		\hline
		& &  \\[-6pt] 
		$N_{\max}$&$b\in\{0,\dots,N_{\max}-1\}$&${(1-\lambda_{k,j})\binom{N_{\max}}{N_{\max}-b}(\frac{1}{\bar{U}_{k,j}^{m}})^{N_{\max}-b}(1-\frac{1}{\bar{U}_{k,j}^{m}})^{b}\atop+\lambda_{k,j}\binom{N_{\max}}{N_{\max}-b+1}(\frac{1}{\bar{U}_{k,j}^{m}})^{N_{\max}-b+1}(1-\frac{1}{\bar{U}_{k,j}^{m}})^{b-1}}$ \\
		\hline
		& &  \\[-6pt] 
		$N_{\max}$&$N_{\max}$&$(1-\lambda_{k,j})(1-\frac{1}{\bar{U}_{k,j}^{m}})^{N_{\max}}+\lambda_{k,j}{N_{\max}}(\frac{1}{\bar{U}_{k,j}^{m}})(1-\frac{1}{\bar{U}_{k,j}^{m}})^{N_{\max}-1}+\lambda_{k,j}(1-\frac{1}{\bar{U}_{k,j}^{m}})^{N_{\max}}$ \\
		\hline
	\end{tabular}
\end{table*}

\begin{Lemma}[Analytical Expression of ${d}_{k,j,m}^{\text{AP}}$] \label{lem:AP}\em

${d}_{k,j}^{\text{AP}}(N_{k,j}^m)$ can be expressed as
\begin{align}\label{eqn:AP}
{d}_{k,j,m}^{\text{AP}}(N_{k,j}^m)&=	\sum_{t=1}^{+\infty}[\mathbf{u}_{k,j,m}(N_{k,j}^m)]^{\mathsf{T}}(\gamma\mathbf{M}_{k,j,m})^{t-1}\mathbf{g}\nonumber\\
&=[\mathbf{u}_{k,j,m}(N_{k,j}^m)]^{\mathsf{T}}(\mathbf{I}-\gamma\mathbf{M}_{k,j,m})^{-1}\mathbf{g},
\end{align}
where the notations of $\mathbf{u}_{k,j,m}(N_{k,j}^m)$, $\mathbf{g}$ and $\mathbf{M}_{k,j,m}$ are defined below.
\begin{itemize}
	
	\item $\mathbf{u}_{k,j,m}(N_{k,j}^m)\in \mathbb{R}^{(N_{\text{max}}+1)\times1}$, whose $N_{k,j}^m$-th entry is $1$ and other entries are all $0$.

	\item $\mathbf{g}\in \mathbb{R}^{(N_{\text{max}}+1)\times1}$, whose $i$-th entry is $i$, $i=0,1,\dots,N_{\max}$.
	\item $\mathbf{M}_{k,j,m}\in \mathbb{R}^{(N_{\text{max}}+1)\times(N_{\text{max}}+1)}$ denotes the transition matrix of the $\{N_{k,j}^m(t)|\forall t\}$, whose entries are given in table \ref{table:M_{k,j,m}}.
\end{itemize}

\end{Lemma}
\begin{proof}
	Please refer to Appendix A.
\end{proof}

\begin{Lemma}[Analytical Expression of ${d}_{m,j}^{\text{ES}}$]\label{lem:ES} \em
${d}_{m,j}^{\text{ES}}(\mathbf{S}_j)$ is given by
\begin{align}
{d}_{m,j}^{\text{ES}}(\mathbf{S}_j)=&[\mathbf{q}_{m,j}(Q_{m,j})]^{\mathsf{T}}\mathbf{c}\nonumber\\
&+\sum_{t=2}^{+\infty}[\mathbf{q}_{m,j}(Q_{m,j})]^{\mathsf{T}}\Big[\gamma\mathbf{P}_{m,j}\Big(\alpha_{m,j}(t-1)\Big)\Big]^{t-1}\mathbf{c},
\end{align}
where the notations of $\mathbf{q}_{m,j}({Q}_{m,j})$,  $\mathbf{c}$, $\mathbf{P}_{m,j}\Big(\alpha_{m,j}(t-1)\Big)$ and $\alpha_{m,j}(t-1)$ are defined below.

	\begin{itemize}
		\item $\mathbf{q}_{m,j}({Q}_{m,j})\in \mathbb{R}^{(L_{\text{max}}\eta_{\text{max}}+1)\times1}$, whose $i$-th entry is 
			\begin{align}
		[\mathbf{q}_{m,j}({Q}_{m,j})]_{i}=\begin{cases}
		1, & i=\eta_{m,j}+(L_{m,j}-1)*\eta_{\text{max}}
		\cr
	0, & 
			i=0,\dots,\eta_{m,j}+(L_{m,j}-1)*\eta_{\text{max}}-1,\atop\eta_{m,j}+(L_{m,j}-1)*\eta_{\text{max}}+1,\dots,\eta_{\text{max}}L_{\text{max}}
		\end{cases}
		\end{align}
		\item $\mathbf{c}\in \mathbb{R}^{(L_{\text{max}}\eta_{\text{max}}+1)\times1}$,  whose $i$-th entry is 
		\begin{align}
		c_{i}=\begin{cases}
		\lceil i/\eta_{\text{max}}\rceil, & i=0,1,\dots,(L_{\text{max}}-1)\eta_{\text{max}}
		\cr
		L_{\text{max}}+\beta, & \text{ otherwise}
		\end{cases}
		\end{align}
		\item $\mathbf{P}_{m,j}\Big(\alpha_{m,j}(t-1)\Big)\in\mathbb{R}^{(L_{\text{max}}\eta_{\text{max}}+1)\times(L_{\text{max}}\eta_{\text{max}}+1)}$ denotes the transition matrix of the $\{Q_{m,j}(t)|\forall t\}$ given the average number of job arrival $\alpha_{m,j,t-1}$, where
		\begin{align}
		\alpha_{m,j}(t)\triangleq\sum_{k\in\mathcal{K}}\mathrm{I}(\omega_{k,j}^{\Pi}=m)\frac{[\mathbf{u}_{k,j,m}(N_{k,j}^m)]^{\mathsf{T}}(\mathbf{M}_{k,j,m})^{t-1}\mathbf{g}}{\bar{U}_{k,j}^m}.
		\end{align}
		The entries of the transition probability matrix $\mathbf{P}_{m,j}(.)$ are provided by table \ref{table:P}.
	\end{itemize}
	
\end{Lemma}
\begin{proof}
	%Please refer to Appendix B.
	The proof is similar to the proof of Lemma \ref{lem:AP}.
\end{proof}

\begin{table*}
	\centering  % 显示位置为中间
	\caption{ENTRIES OF MATRIX $[\mathbf{P}_{m,j}(\alpha_{m,j}(t-1))]$}  % 表格标题 
	\label{table:P}  % 用于索引表格的标签
	%字母的个数对应列数，|代表分割线
	% l代表左对齐，c代表居中，r代表右对齐
	\begin{tabular}{|c|c|c|}  
		\hline  % 表格的横线
		& & \\[-6pt]  %可以避免文字偏上来调整文字与上边界的距离
		$q$&$p$&$[\mathbf{P}_{m,j}(\alpha_{m,j}(t-1))]_{q,p}$\\  % 表格中的内容，用&分开，\\表示下一行
		\hline
		& & \\[-6pt]  %可以避免文字偏上来调整文字与上边界的距离
		$0$&$0$&$1-\alpha_{m,j}(t-1)$\\  %
		\hline
		& & \\[-6pt]  %可以避免文字偏上来调整文字与上边界的距离
		$0$&$b\in\{1,\dots,N_{\text{max}}\}$&$(1-\alpha_{m,j}(t-1))f_{m,j}(b)$\\  %
		\hline
			& & \\[-6pt]  %可以避免文字偏上来调整文字与上边界的距离
		$0$&$\{N_{\text{max}}+1,\dots,Q_{\text{max}}N_{\text{max}}\}$&$0$\\  %
		\hline
			& & \\[-6pt]  %可以避免文字偏上来调整文字与上边界的距离
		$a\in\{N\times N_{\text{max}}+2,\dots,N\times N_{\text{max}}+N_{\text{max}}|N=0,1,\dots,Q_{\text{max}}-1\}$&$a-1$&$1-\alpha_{m,j}(t-1)$\\  %
		\hline
		& & \\[-6pt]  %可以避免文字偏上来调整文字与上边界的距离
		$a\in\{N\times N_{\text{max}}+2,\dots,N\times N_{\text{max}}+N_{\text{max}}|N=0,1,\dots,Q_{\text{max}}-1\}$&$a+N_{\text{max}}-1$&$\alpha_{m,j}(t-1)$\\  %
		\hline
		& & \\[-6pt]  %可以避免文字偏上来调整文字与上边界的距离
		$a\in\{N\times N_{\text{max}}+2,\dots,N\times N_{\text{max}}+N_{\text{max}}|N=0,1,\dots,Q_{\text{max}}-1\}$&${\{1,\dots,N_{\text{max}}Q_{\text{max}}\}\atop\setminus\{a-1,a+N_{\text{max}}-1\}}$&$0$\\  %
		\hline
			& & \\[-6pt]  %可以避免文字偏上来调整文字与上边界的距离
		$a\in\{N\times N_{\text{max}}+1|N=0,1,\dots,Q_{\text{max}}-1\}$&$b\in\{a-N_{\text{max}},\dots,a-1\}$&${(1-\alpha_{m,j}(t-1))\atop\times f_{m,j}(b-a+1+N_{\text{max}})}$\\  %
		\hline
			& & \\[-6pt]  %可以避免文字偏上来调整文字与上边界的距离
		$a\in\{N\times N_{\text{max}}+1|N=0,1,\dots,Q_{\text{max}}-1\}$&$b\in\{a,\dots,a-1+N_{\text{max}}\}$&${(1-\alpha_{m,j}(t-1))\atop\times f_{m,j}(b-a+1)}$\\  %
		\hline
			& & \\[-6pt]  %可以避免文字偏上来调整文字与上边界的距离
		$a\in\{N\times N_{\text{max}}+1|N=0,1,\dots,Q_{\text{max}}-1\}$&${\{1,\dots,N_{\text{max}}Q_{\text{max}}\}\setminus\atop\{\{a-N_{\text{max}},\dots,a-1\}\cup \{a,\dots,a-1+N_{\text{max}}\}\}}$&$0$\\  %
		\hline
	\end{tabular}
\end{table*}

\begin{figure*}
	\centering
	\includegraphics[width=1\textwidth]{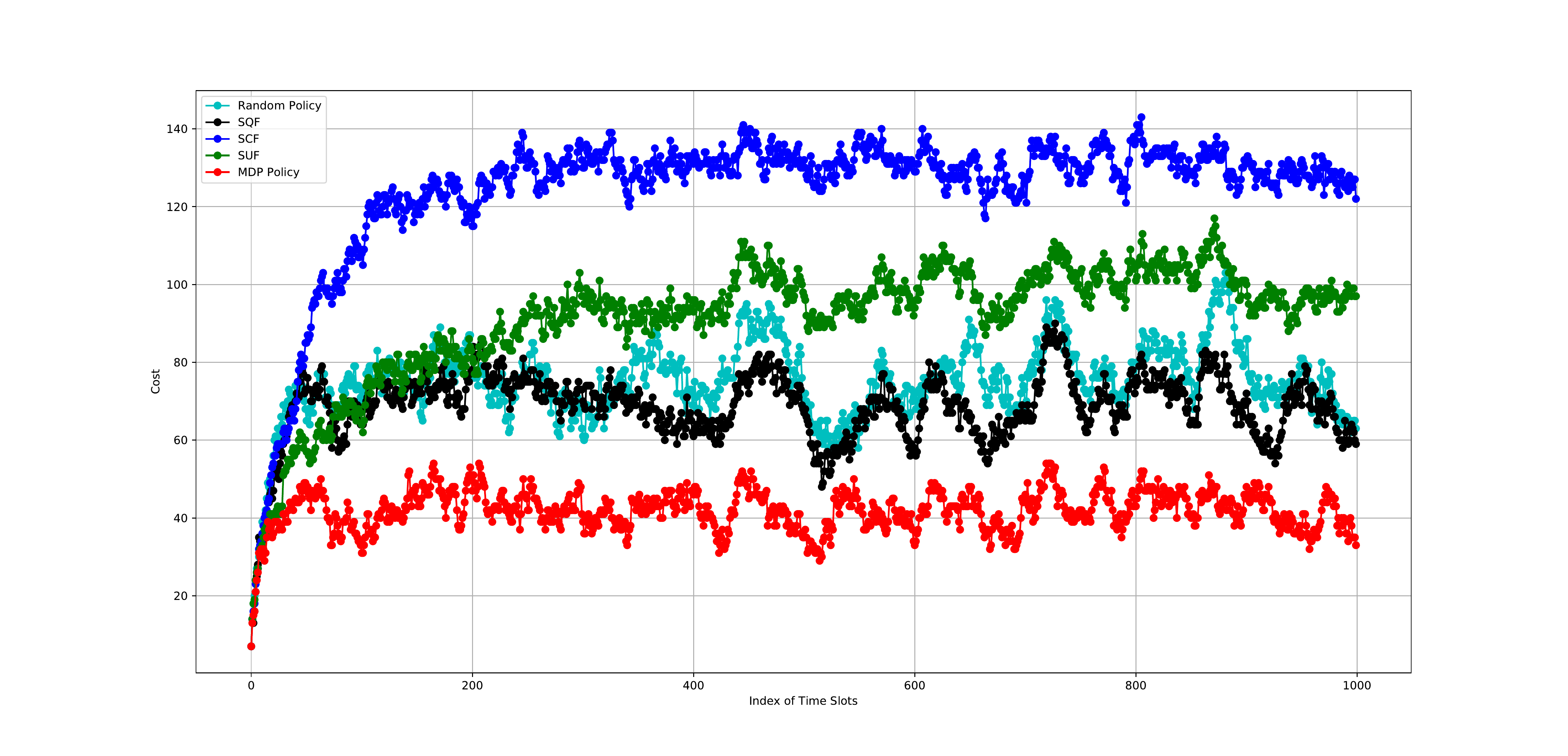}
	\caption{The cost versus time slots when the average uploading delays and the job computation times are comparable. For example, the average uploading delay of the first job type from $1$-st AP to $1$-st edge server is $10$ time slots, and the job computation time of the first job type at the $1$-st edge server ranges from $10$ to $15$ time slots.}
	\label{fig:timeline}
\end{figure*}

\subsection{Scheduling Policy with One-Step Policy Iteration}
In this part, we use the value function of the baseline policy $\{V_{\Pi}(\mathbf{S})|\forall \mathbf{S}\}$ derived in the previous part to approximate the value function of the optimal policy $\{{V}(\mathbf{S})|\forall \mathbf{S}\}$ in \eqref{eqn:bellman}, and derive the proposed scheduling policy. Because the expression of value function $\{V_{\Pi}(\mathbf{S})|\forall \mathbf{S}\}$ is provided, the value iteration
can be avoid, which significantly reduces the computation complexity. 
Note that $V_{\Pi}(\mathbf{S})$, $g(\mathbf{S}(t),\Omega(\mathbf{S}(t)))$ and $\Pr\Big(\mathbf{S}(t+1)\Big|\mathbf{S}(t),\Omega(\mathbf{S}(t))\Big) $ in \eqref{eqn:decouple} can be decoupled for each type of job, Problem \ref{Pro:main} with the value function approximation can be decoupled into the following per-type optimization.

\begin{Problem}[Sub-Optimal Scheduling Problem of $j$-th Type]\label{Pro:sub-opt} \em
	\begin{align}
	\min_{\Omega_j(t)}&\bigg[g_j\bigg(\mathbf{S}_j(t),\Omega_j(\mathbf{S}_j(t))\bigg)\nonumber\\
	&+\gamma\!\!\!\sum_{\mathbf{S}_j(t+1)}\!\!\!\Pr\bigg(\mathbf{S}_j(t+1)\bigg|\mathbf{S}_j(t),\Omega_j(\mathbf{S}_j(t))\bigg) W_{j}(\mathbf{S}_j(t+1))\bigg],
	\end{align}
where $W_{j}(.)$ is defined in \eqref{eqn:V_baseline}.
\end{Problem}

 Problem \ref{Pro:sub-opt} is NP-hard due to the combinatorial search of the computation edge servers, and it is difficult to find the optimal solution. Hence, instead of the intractable optimal solution, we propose a sub-optimal low-complexity solution as follows.

\begin{Algorithm}[Proposed Scheduling Policy] \label{alg:proposed}
With the system state of $j$-th type jobs $\mathbf{S}_{j}$, the proposed scheduling policy $\Pi_j^{\dagger}(\mathbf{S}_{j})$ is given below.\em
	\begin{itemize}
		\item  {\bf Step 1: }Let $\ell=0$. Initialize dispatching action with the $\Pi_j^{\ell}= \{\omega_{k,j}^{\ell}=\omega_{k,j}^{\Pi}|\forall k\}$ and let $X=W_j(\mathbf{S}_j(t))$.
		\item {\bf Step 2: }Let $\ell=\ell+1$ and update the set of dispatching action from $\Pi_{j}^{l-1}$ to $\Pi_{j}^l$ as 
		$\omega_{k,j}^{\ell} = \omega_{k,j}^{\ell-1}, \forall \ell \neq l$, and $\omega_{\ell,j}^{\ell}$ is the solution of the following optimization problem. 
		\begin{align}	
	Y_{\ell}=&\min_{\omega_{\ell,j}\in\mathcal{M}}\bigg[g_j\bigg(\mathbf{S}_j(t),\{\omega_{k,j}^{\ell}|\forall k\neq \ell\}\cup\{\omega_{\ell,j}\} \bigg)\nonumber\\
		&+\gamma\!\!\!\sum_{\mathbf{S}_j(t+1)}\!\!\!\Pr\bigg(\mathbf{S}_j(t+1)\bigg|\mathbf{S}_j(t),\{\omega_{k,j}^{\ell}|\forall k\neq \ell\}\cup\{\omega_{\ell,j}\}\bigg)\nonumber\\
		&\quad\quad\quad\quad\quad\times W_{j}(\mathbf{S}_j(t+1))\bigg].
		\end{align}
		If $Y_{\ell}<X$, let $X=Y_{\ell}$.

   \item {\bf Step 3:} If $\ell=K$, algorithm terminates. The proposed scheduling policy is $\Pi_{j}^{\dagger}(\mathbf{S}_j(t))=\Pi_{j}^{K}(\mathbf{S}_j(t))$. Otherwise, go to Step 2.

		\end{itemize}
\end{Algorithm}

\begin{figure*}
	\centering
	\includegraphics[width=1\textwidth]{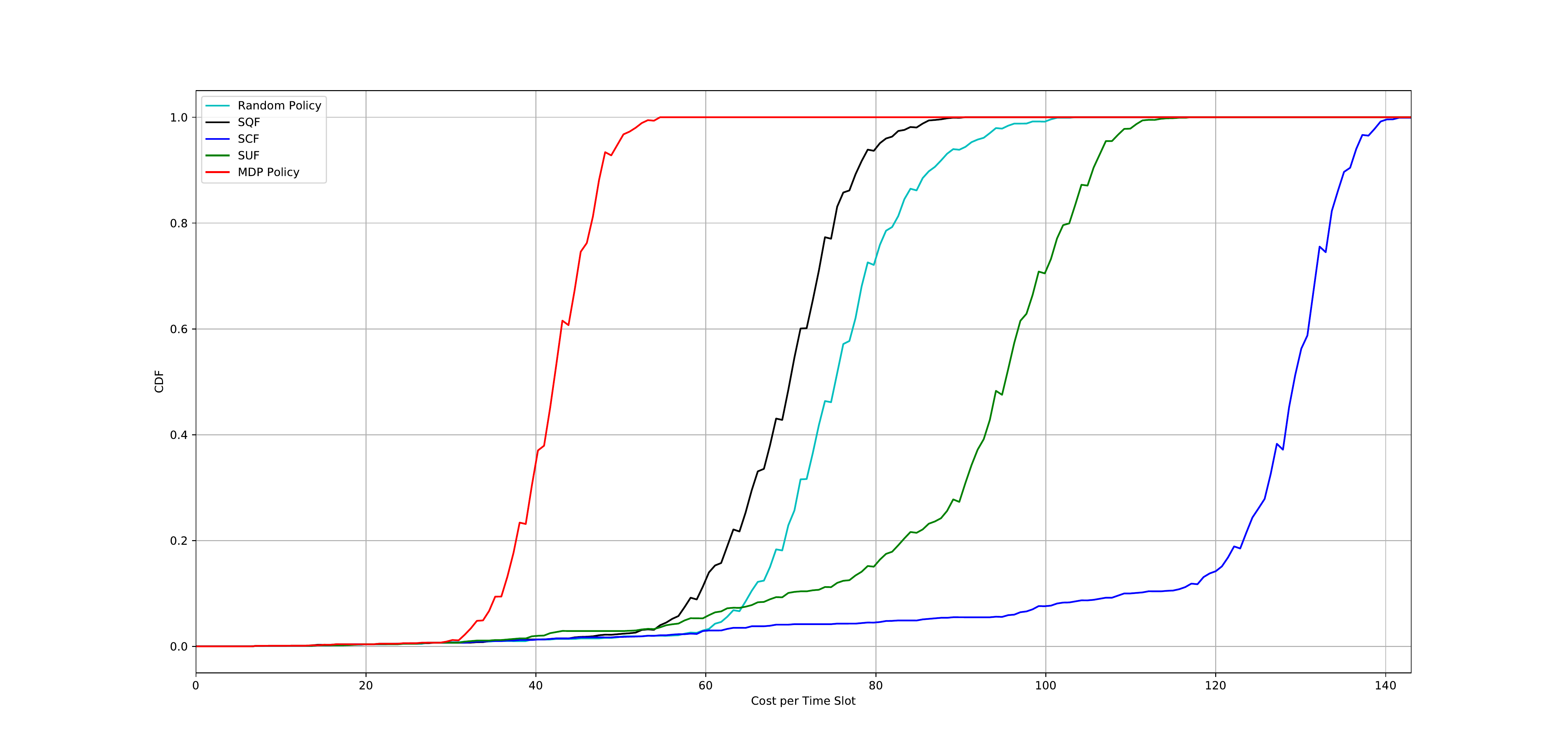}
	\caption{Cumulative distribution function(CDF) of the cost per time slot when average uploading delays and job computation times are comparable. For example, the average uploading delay of the first job type from first AP to first edge server is $10$ time slots, and the job computation time of the first job type at the first edge server ranges from $10$ to $15$ time slots.}
	\label{fig:cdf1}
\end{figure*}

The complexity of Algorithm \ref{alg:proposed} is $O(KM)$.
 Although it is sub-optimal solution of Problem \ref{Pro:sub-opt}, its performance is superior to the baseline policy, which is summarized in the following lemma.

\begin{Lemma}[Performance Bound] \em Let $V_{{\Pi}^{\dagger}}(.)$ be the value function of the policy ${\Pi}^{\dagger}\triangleq\{\Pi_{j}^{\dagger}|\forall j\in \mathcal{J}\}$, i.e.,
\begin{align}
	&V_{{\Pi}^{\dagger}}(\mathbf{S})\nonumber\\
	&\quad \triangleq \lim\limits_{T\to \infty}\mathbb{E}_{\{\mathbf{S}(t)|\forall t\}}^{{\Pi}^{\dagger}}\bigg[\sum_{t=1}^{T} \gamma^{t-1}g\bigg(\mathbf{S}(t),{\Pi}^{\dagger}(\mathbf{S}(t))\bigg)\bigg|\mathbf{S}(1)=\mathbf{S} \bigg],
\end{align}
 we have
\begin{align}
	V(\mathbf{S})\leq 	V_{{\Pi}^{\dagger}}(\mathbf{S})\leq	V_{{\Pi}}(\mathbf{S}), \forall \mathbf{S}.
\end{align}
\end{Lemma}

\begin{figure*}
	\centering
	\includegraphics[width=1\textwidth]{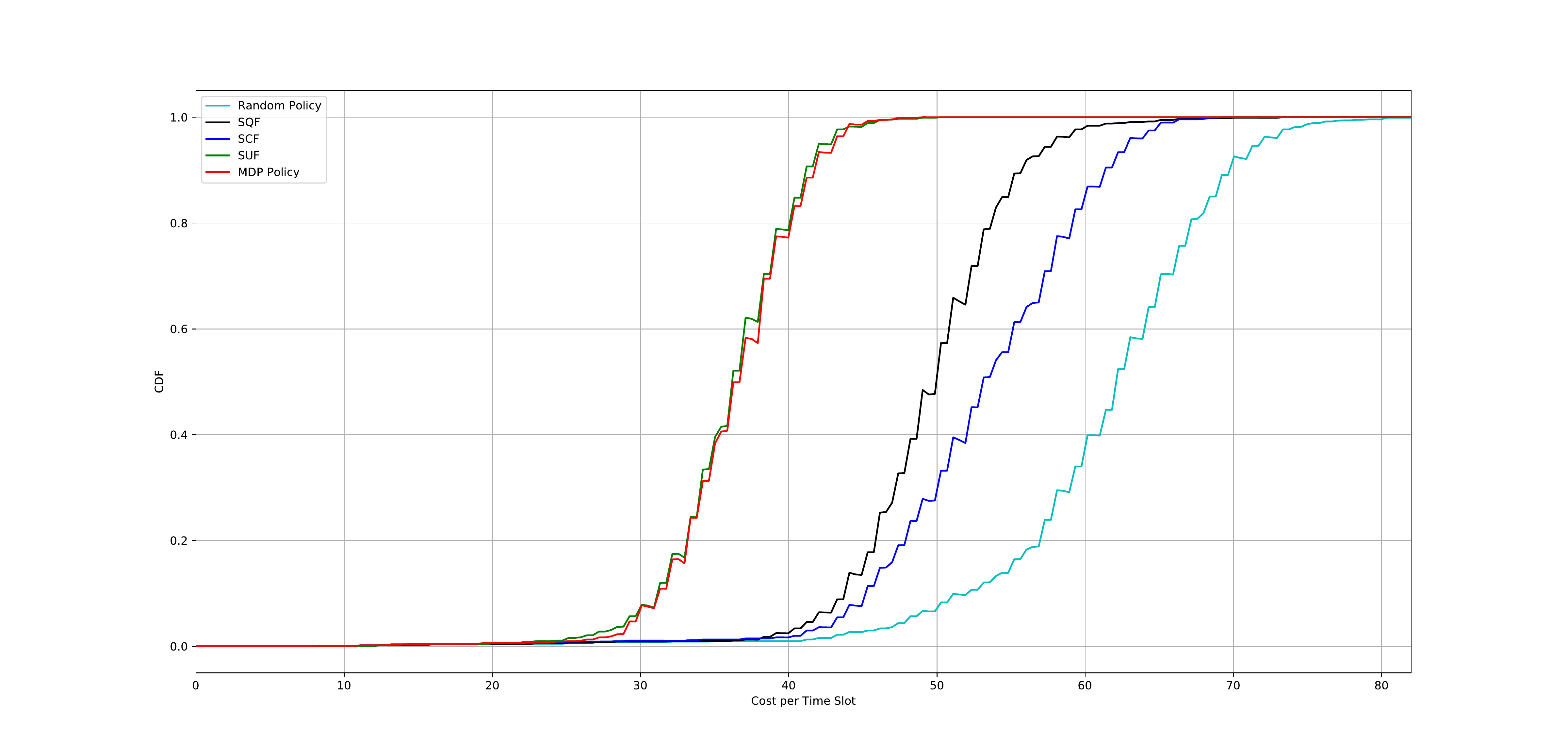}
	\caption{Cumulative distribution function (CDF) of the cost per time slot when average uploading delays are dominant. For example, the average uploading delay of the first job type from first AP to first edge server is $10$ time slots, and the average job computation time of the first job type at the first edge server is $1$ time slot.}
	\label{fig:cdf2}
\end{figure*}

\begin{figure*}
	\centering
	\includegraphics[width=1\textwidth]{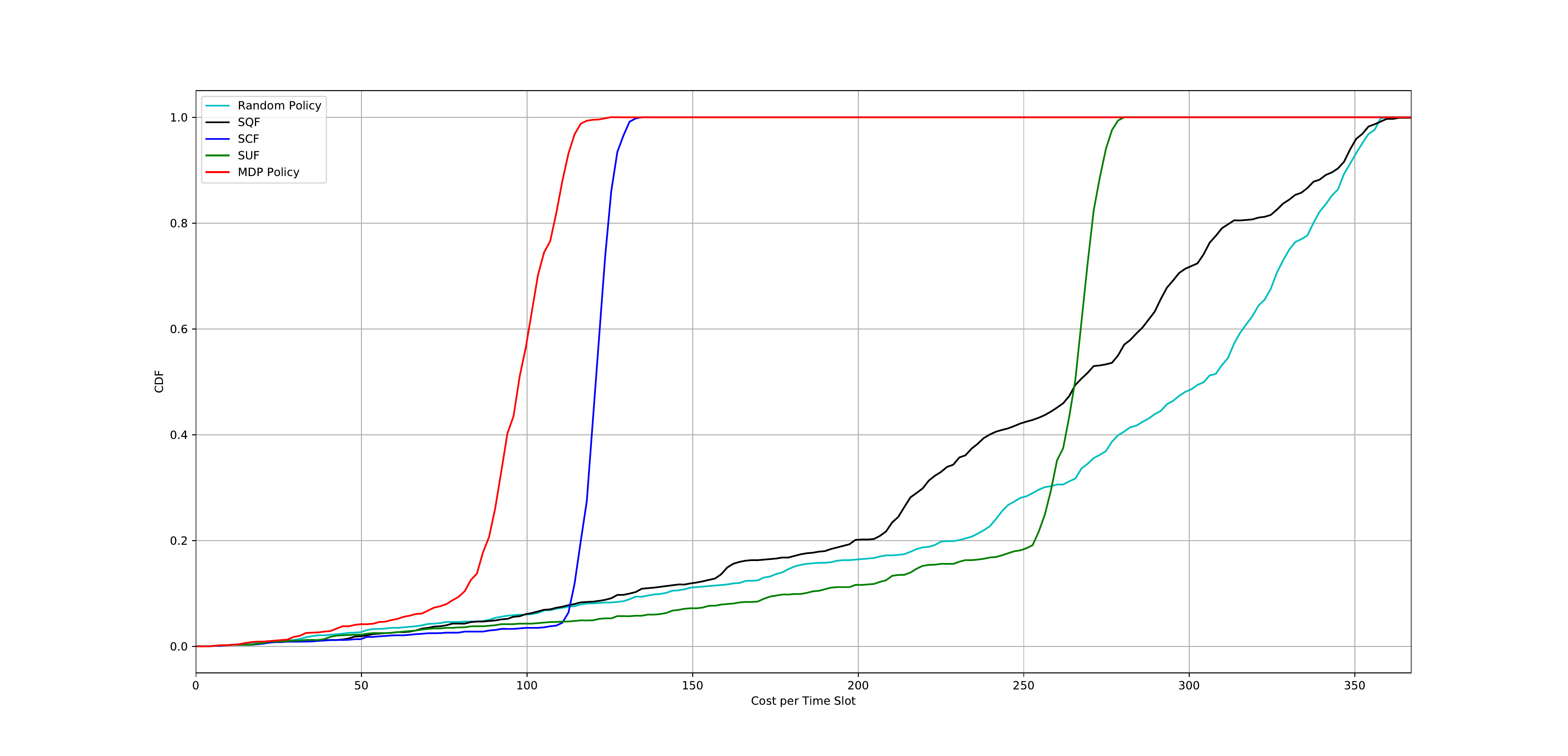}
	\caption{Cumulative distribution function (CDF) of the cost per time slot when average job computation times are dominant. For example, the average uploading delay of the first job type from first AP to first edge server is $1$ time slots, and the average job computation time of the first job type at the first edge server ranges from $10$ to $15$ time slots.}
	\label{fig:cdf3}
\end{figure*}

\begin{proof}
	Since policy $\Pi^{\dagger}$ is not optimal policy, $	V(\mathbf{S})\leq V_{{\Pi}^{\dagger}}(\mathbf{S})$ is straightforward.  Due to the {\em Policy Improvement Property} in \cite{DP}, we have $V_{{\Pi}^{\dagger}}(\mathbf{S})\leq	V_{{\Pi}}(\mathbf{S})$.
\end{proof}

	To the best of our knowledge, the performance can hardly be bounded in the existing approximate MDP methods.
 This paper shows a low complexity
approximate MDP method whose performance can be bounded analytically. 
\begin{Remark}[Complexity Analysis] \em
 In the above approximation approach, the computation complexity of
	value function calculation is $O(J(KM+M))$. 
	On the other hand, the optimal solution of MDP suffers from the curse of dimensionality. Specifically, the computation complexity of the conventional value
	iteration algorithm is $O\Big((N_{\text{max}})^{2KJM}(L_{\text{max}}\eta_{\text{max}})^{2MJ}M^{KJ}\Big)$ and the memory requirement is $O\Big((N_{\text{max}})^{KJM}(L_{\text{max}}\eta_{\text{max}})^{MJ}\Big).$
\end{Remark}

\section{Numerical Simulations}

In this section, we evaluate the performance of the proposed
low-complexity sub-optimal scheduling policy (Algorithm \ref{alg:proposed}) by numerical
simulations. In the simulation, there are $5$ APs, $3$ edge servers and $10$ types of jobs in the network. The computation time of each job type at each edge server follows the uniform distribution.
The following four benchmark schemes
are compared with the proposed scheduling scheme.

	\begin{itemize}
		\item SQF (shortest queue first) algorithm: APs dispatch the jobs to the edge server with the shortest queue length of the same type;
		\item SUF (shortest uploading time first) algorithm: APs dispatch the jobs to the edge server with shortest \emph{expected uploading time};
		\item SCF (shortest computation time first) algorithm: APs dispatch the jobs to the edge server with shortest \emph{expected computation time} for that job type;
		\item Random algorithm: APs randomly dispatch the jobs to the edge server at each time slot.
	\end{itemize}

Moreover, in the proposed scheme, we use the SCF algorithm as the baseline policy \ref{Pol:baseline}. In Fig. \ref{fig:timeline} and Fig. \ref{fig:cdf1}, the performance of the five schemes are compared when average uploading delays and job computation times are comparable. It can be observed that the proposed algorithm has significantly less cost per time slot than all the benchmarks. Note that in SQF algorithm, the jobs dispatching can be adjusted according to system state; whereas, the SUF and SCF algorithms have fixed jobs dispatching action in all the time slots. SQF algorithm has better performance than SUF and SCF algorithms.

 In Fig. \ref{fig:cdf2}, the average uploading delays are dominant, compared with the average job computation times. The performance of proposed algorithm is almost the same as the performance of SUF algorithm. Hence, in the edge computing network with dominant uploading delays, APs tend to dispatch the jobs to the edge server with shortest \emph{expected uploading time}.
 
 In Fig. \ref{fig:cdf3}, the average job computation times are dominant, compared with the average uploading delays. It can be observed that the proposed algorithm has less cost per time slot than all the benchmarks. Note that SCF algorithm has better performance than other benchmarks, the APs tend to dispatch the jobs to the edge server with shortest \emph{expected computation time} in this situation.

\section{Conclusions}
In this paper, we consider the cooperative jobs dispatching in an edge computing network with multiple APs and edge servers. The job uploading delay and computation time are both random and unpredictable. We formulate the joint optimization of jobs dispatching at all the APs and all the time slots as an infinite-horizon MDP with discounted cost. In order to avoid the curse of dimensionality, we also introduce a low-complexity sub-optimal solution based on one-step policy iteration from a baseline policy. The analytical performance bound is derived. Finally, it is shown by simulations that our proposed scheme has better performance than various benchmarks.

As the future work, we shall extend the proposed algorithm to one new scenario where the delay of collecting complete system state information at each AP is not negligible. Moreover, the memoryless distribution of uploading delay can also be generalized to arbitrary distribution.

\section*{Appendix}

\noindent\textbf{A)}\quad PROOF OF LEMMA \ref{lem:AP}

The entries of matrix $\mathbf{M}_{k,j,m}$ is
\begin{align*}
[\mathbf{M}_{k,j,m}]_{q,p}\triangleq\Pr\bigg(N_{k,j}^m(t+1)=p\bigg|N_{k,j}^m(t)=q,\Pi\bigg).
\end{align*}
Then, we have following discussion on $[\mathbf{M}_{k,j,m}]$.
\begin{itemize}
	\item $q=0$, $p=0$: There are no $j$-th type job arriving at the $k$-th AP. Hence $[\mathbf{M}_{k,j,m}]_{0,0}=1-\lambda_{k,j}$.
	\item $q=0$, $p=1$: There is one $j$-th type job arriving at the $k$-th AP. Hence $[\mathbf{M}_{k,j,m}]_{0,1}=\lambda_{k,j}$.
	\item $q=a\in\{1,\dots,N_{\text{max}}-1\}$, $p=b\in\{0,\dots,a\}$: (i)
	There are no $j$-th type job arriving at the $k$-th AP and $D_{k,j}^{m}(t)=(a-b)$; (ii) There is one $j$-th type job arriving at the $k$-th AP and $D_{k,j}^{m}(t)=(a-b)-1$.
    Hence,  $[\mathbf{M}_{k,j,m}]_{a,b}=(1-\lambda_{k,j})\binom{a}{a-b}(\frac{1}{\bar{U}_{k,j}^{m}})^{a-b}(1-\frac{1}{\bar{U}_{k,j}^{m}})^{b}+\lambda_{k,j}\binom{a}{a-b+1}(\frac{1}{\bar{U}_{k,j}^{m}})^{a-b+1}(1-\frac{1}{\bar{U}_{k,j}^{m}})^{b-1}$.
    \item $ q=a\in\{1,\dots,N_{\text{max}}-1\}$, $p=a+1$: There is one $j$-th type job arriving at the $k$-th AP and  $D_{k,j}^{m}(t)=0$. Hence, $[\mathbf{M}_{k,j,m}]_{a,a+1}=\lambda_{k,j}(1-\frac{1}{\bar{U}_{k,j}^{m}})^{a}$.
    \item  $q=N_{\max}$, $p=b\in\{0,\dots,N_{\max}-1\}$: (i)
    There are no $j$-th type job arriving at the $k$-th AP and $D_{k,j}^{m}(t)=(N_{\text{max}}-b)$; (ii) There is one $j$-th type job arriving at the $k$-th AP and $D_{k,j}^{m}(t)=(N_{\text{max}}-b)-1$. Hence, $[\mathbf{M}_{k,j,m}]_{N_{\text{max}},b}=(1-\lambda_{k,j})\binom{N_{\max}}{N_{\max}-b}(\frac{1}{\bar{U}_{k,j}^{m}})^{N_{\max}-b}(1-\frac{1}{\bar{U}_{k,j}^{m}})^{b}+\lambda_{k,j}\binom{N_{\max}}{N_{\max}-b+1}(\frac{1}{\bar{U}_{k,j}^{m}})^{N_{\max}-b+1}(1-\frac{1}{\bar{U}_{k,j}^{m}})^{b-1}$.
    \item  $q=N_{\max}$, $p=N_{\max}$: (i)There are no $j$-th type job arriving at the $k$-th AP and $D_{k,j}^{m}(t)=0$; (ii) There is one $j$-th type job arriving at the $k$-th AP and $D_{k,j}^{m}(t)=1$; There is one $j$-th type job arriving at the $k$-th AP and $D_{k,j}^{m}(t)=0$. Hence, $[\mathbf{M}_{k,j,m}]_{N_{\text{max}},N_{\text{max}}}=(1-\lambda_{k,j})(1-\frac{1}{\bar{U}_{k,j}^{m}})^{N_{\max}}+\lambda_{k,j}{N_{\max}}(\frac{1}{\bar{U}_{k,j}^{m}})(1-\frac{1}{\bar{U}_{k,j}^{m}})^{N_{\max}-1}+\lambda_{k,j}(1-\frac{1}{\bar{U}_{k,j}^{m}})^{N_{\max}}$.
    \item Otherwise, $[\mathbf{M}_{k,j,m}]_{q,p}=0$.
\end{itemize}
    \spaceblank
   To prove the second equity of equation \eqref{eqn:AP}, we first show $||\gamma\mathbf{M}_{k,j,m}||<1$, where $||.||$ is the matrix norm. It clear that $||\gamma\mathbf{M}_{k,j,m}||=\gamma\rho(\mathbf{M}_{k,j,m})$, where $\rho(\mathbf{M}_{k,j,m})$ is the spectrum radius of  $\mathbf{M}_{k,j,m}$. According to Perron-Frobenius Theorem \cite{matrix},  the spectrum radius of transition probability matrix is $1$. Since $\mathbf{M}_{k,j,m}$ is transition probability matrix, we have $||\gamma\mathbf{M}_{k,j,m}||=\gamma<1$. Let $\mathbf{X}_n=\sum_{t=1}^{n}(\gamma\mathbf{M}_{k,j,m})^{t-1}$, we have
   \begin{align*}
   \mathbf{X}_n=(\mathbf{I}-\gamma\mathbf{M}_{k,j,m})^{-1}-(\gamma\mathbf{M}_{k,j,m})^{n+1}(\mathbf{I}-\gamma\mathbf{M}_{k,j,m})^{-1}.
   \end{align*}
Then
\begin{align*}
\lim\limits_{n\to+\infty}  \mathbf{X}_n=(\mathbf{I}-\gamma\mathbf{M}_{k,j,m})^{-1}.
\end{align*}
Hence, the \eqref{eqn:AP} is straightforward. 
\spaceblank

%\noindent\textbf{B)}\quad PROOF OF LEMMA \ref{lem:ES}

%The entries of matrix $\mathbf{P}_{m,j}$ is

%%%%%%%%%%%%%%%%%  references  %%%%%%%%%%%%%%%%%%%%%%%%%%%%%%%%%%%%%%%%%%%%%%%%

%\section*{About the Authors}\footnotesize\vskip 2mm

%\parpic{\includegraphics[width=22.86mm,height=32mm]{photo1.eps}}%
%\noindent{\bf Firstname1 Lastname 1} ********* (Email: ***********)
%\vskip4.59mm

%\parpic{\includegraphics[width=22.86mm,height=32mm]{photo2.eps}}%
%\noindent{\bf Firstname2 Lastname2} [corresponding author] ************ (Email: ***********)

\end{document}